\documentclass[reqno, 11pt]{article}

\usepackage[dvips]{graphicx}
\usepackage{ifpdf}
\usepackage{latexsym}
\usepackage{graphics}
\usepackage{graphicx}
\usepackage{multirow}
\usepackage{amsmath}
\usepackage{amssymb}
\usepackage{amsthm}
\usepackage{enumerate}
\usepackage{graphicx,color,amssymb,latexsym,amsfonts,amsmath}

\newtheorem{theorem}{Theorem}[section]
\newtheorem{proposition}[theorem]{Proposition}

\newtheorem{lemma}[theorem]{Lemma}
\newtheorem{remark}[theorem]{Remark}
\vspace{-10mm}

\begin{document}

\title {Option Pricing for Symmetric L\'evy Returns with Applications}

\author{\hspace{0mm}$\begin{array}{l}\mbox{Kais Hamza, Fima C. Klebaner, Zinoviy Landsman$^{\dag}$ and Ying-Oon Tan} \vspace{2mm} \\
\mbox{\em School of Mathematical Sciences, Monash University,}\\
\mbox{\em Melbourne, Clayton Australia.} \vspace{2mm} \\
\mbox{\em $^\dag$Department of Statistics, University of Haifa,}\\
\mbox{\em Haifa, Israel.}\end{array}$}
\date{}
\maketitle

\begin{center} \bf Abstract \end{center}

This paper considers options pricing when the assumption of normality is replaced with that of the symmetry of the underlying distribution.
Such a market affords many equivalent martingale measures (EMM). However we argue (as in the discrete-time setting of \cite{KlebanerLandsman07})
that an EMM that keeps distributions within the same family is a ``natural'' choice.
We  obtain Black-Scholes type option pricing formulae for symmetric
Variance-Gamma and symmetric Normal Inverse Gaussian models.

 \vskip5mm

{\bf Keywords}: Symmetric distribution, L\'evy processes, equivalent
martingale measure, risk-neutral pricing, option pricing, Variance
Gamma process, Normal Inverse Gaussian process.

\vskip5mm

{\bf AMS 2000 subject classification}: 60G51, 60E99, 91B24


\section{Introduction}

In the classical Black-Scholes model  the stock price follows a
geometric Brownian motion and the return process is a Brownian motion
with drift.

In some cases empirical evidence shows that a more general
symmetric distribution is more appropriate for returns -- see for example
\cite{Mendelbrot63}, \cite{Mendelbrot67}, \cite{Fama65},
\cite{BlattbergGonedes74}, \cite{Hurlimann95}, \cite{Hurlimann01}.

In this work we replace the assumption of normality with that of symmetry,
while retaining all other assumptions such as independence and
stationarity of increments. This leads to returns that are symmetric
L\'evy processes; such stock prices are known as log-symmetric L\'evy processes.

We adopt a classical approach to the definition of symmetry; a random variable $Y$
has a symmetric distribution if for some $\mu$, the location parameter,
$(Y - \mu)$ and $-(Y - \mu)$ have the same distribution. In turn, a symmetric L\'evy
process is defined as having symmetric marginal distributions. This is easily shown
to reduce to assuming that the L\'evy measure is symmetric (about zero):
$\nu(A)=\nu(-A)$ for any Borel set $A\subset \mathbb{R}$.

Our definition of symmetry differs from that in \cite{FajardoMordecki06} where a L\'evy market
is said to be symmetric if a certain law before and
after the change of measure through Girsanov's theorem coincide.

The literature on option pricing with L\'evy processes is vast -- see for example,
\cite{Benhamou02}, \cite{CarrWu04}, \cite{Chan99}, \cite{FajardoMordecki06},
\cite{MadanCarrChang98}, \cite{Schoutens03}.
In particular, it is well known that L\'evy market models, barring the Brownian motion case, are
incomplete (\cite{Schoutens03}, p.77) and the choice of an EMM is not unique. In fact there are
infinitely many EMMs to choose from and any selection is arbitrary and motivated by various other
considerations.  Among the most popular methods are the Esscher transform (\cite{GerberShiu94},
\cite{KallsenShiryaev02}, \cite{Chan99}), minimum entropy martingale measure (\cite{Frittelli00},
\cite{FujiwaraMiyahara03}, \cite{Miyahara99}), minimal martingale measure (\cite{FollmerSchweizer91},
\cite{Chan99}), minimax and minimal distance martingale measure \cite{GollRuschendorf01},
variance-optimal martingale measure \cite{Schweizer96}, and mean-correcting martingale measure
(\cite{Schoutens03}, chapter 6).

In some cases the Esscher transform produces a continuum of EMMs that require further refinement on the selection by optimizing the relative entropy or some other utility function, \cite{KuchlerTappe08}. In reality, it is hard to tell which measure the market chooses, and this topic requires further research.

The choice of EMM is important not only for obtaining the price of an option but also for calculating
hedging parameters. By the change of num\'eraire formula, these parameters give
probabilities of option exercise  under different EMMs.

In the case of L\'evy processes with symmetric marginal distributions there is a unique EMM within the same family of distributions as the real world distribution.  If the process has a Brownian component, then the natural EMM is the same as in the classical case, obtained by changing the drift (location) parameter. If the process does not have a a Brownian component (and $\mu<r$ -- see Section \ref{OptionPricing}), then the natural EMM is obtained by changing the variance (scale) parameter.
In both cases, we obtain closed form option pricing formulae in which the normal distribution is replaced by other symmetric distributions.
This is reminiscent of the suggestion made by McDonald \cite{McDonald96} in 1996, with the difference that in that paper arbitrage is possible whereas our approach is arbitrage-free.

The search for a ``natural'' EMM under symmetry started in \cite{KlebanerLandsman07} in a discrete-time setting. In this paper we extend the
exploration to continuous-time models. The main contributions of the present work can be summarized as follows.
\begin{itemize}
\item The model for the stock price process is $S_t = S_0e^{Y_t}$, where $Y_t$ is a symmetric L\'evy process.
\item As a L\'evy process, $Y_t$ is specified by the characteristic triplet $(\mu, c, \nu)$. As a random variable, it is described
by the parameters $(\mu,\sigma^2,\psi)$ of the symmetric family of $Y_1$. We show how the two characterizations relate to each other.
\item We construct an equivalent measure $\mathbb{Q}$ under which
\begin{enumerate}
\item[(1)] the symmetric L\'evy process $Y_t$ remains a symmetric L\'evy process;
\item[(2)] the distribution of the L\'evy process $Y_t$ remains in the same symmetric family of distributions as the real world distribution;
\item[(3)] the discounted price process $e^{-rt}S_t$ is a martingale.
\end{enumerate}
We call such a change of measure a natural equivalent martingale measure.
\item  We derive option pricing formulae under the natural EMM.
\end{itemize}

A brief account of L\'evy processes and symmetric
 distributions necessary for our purposes are given in Section 2.
 In Section 3, we give the construction of a natural equivalent martingale measure for
 log-symmetric L\'evy processes.
In Section 4 we consider the option pricing with a natural EMM. Section
5 contains applications of this approach to symmetric Variance Gamma
and Normal Inverse Gaussian models.


\section{Preliminaries}
\subsection{L\'evy Processes with Symmetric Marginal Distributions}
A L\'evy process $(Y_t)_{t\geq 0}$ on $\mathbb{R}$ is a process with independent and
stationary increments. It is defined on a probability space
$(\Omega, \mathcal{F}, \mathbb{P})$ endowed with a complete filtration
$\{\mathcal{F}_t\}_{t \geq 0}$ to which $Y_t$ is adapted.
$Y_t$ has right-continuous with left limits sample paths, and $Y_t -
Y_s$ is independent of $\mathcal{F}_s$ and has the same distribution
as $Y_{t-s}$ for $0\leq s<t$. A L\'evy process is fully
determined by its initial value, $Y_0$, here assumed to be nil, and the
distribution of the increment  over one unit time interval, $Y_1$.
The distribution of   $Y_t$ is infinitely divisible for any $t$, and
its characteristic function   satisfies
    \begin{align}
    \label{InftyDevisible}
     \mathbb{E}(e^{iuY_t}) = \big(\mathbb{E}\left[e^{iuY_1}\right]\big)^t,\quad u \in \mathbb{R}.
    \end{align}
By the L\'{e}vy-Khintchine representation,
    \begin{align}
    \mathbb{E} \left[e^{iuY_1}\right] = e^{\Lambda(u)}, \label{eq:LevyKhintchine}
    \end{align}
with the characteristic exponent
    \begin{align}
    \Lambda(u) = i\mu u -\frac{1}{2}c^2u^2 + \int_{\mathbb{R}}\big(e^{iuy}-1-iuy1_{\{|y| \leq 1\}}\big) \nu(dy), \label{eq:LevyKhintchine2}
    \end{align}
where $\mu \in \mathbb{R}$, and $\nu$ is  a L\'evy measure satisfying
$\nu(\{0\})=0$ and $\int_{\mathbb{R}}(1 \wedge y^2)\nu(dy) <
\infty$. The triplet $(\mu, c, \nu)$ is referred to as the
characteristic triplet of $Y$.

We call a L\'evy process $(Y_t)_{t\geq 0}$ symmetric if, for each $t\geq0$, the random variable $Y_t$ is symmetric (about the location parameter $\mu t$):
$(Y_t-\mu t)$ and $(\mu t-Y_t)$ have the same distribution. By \eqref{eq:LevyKhintchine}, this is easily seen to be equivalent to the random variable $Y_1$ being symmetric (about $\mu$), and by \eqref{eq:LevyKhintchine2}, to the L\'evy measure $\nu$ being symmetric (about 0): for any Borel set $A\subset \mathbb{R}$,
$\nu(-A) = \nu(A)$, where $-A=\{x\in \mathbb{R}: -x\in A\}$. In this case, the characteristic exponent $\Lambda$ can be written as (e.g. \cite{Sato1999}, p.263)
\begin{equation}
\Lambda(u) = iu\mu -\frac{1}{2}c^2u^2 - 2\int_{0}^{\infty}(1 - \cos (uy)) \nu(dy). \label{eq:LevyKhintchineSymm}
\end{equation}

In what follows we assume that $Y_1$ has finite mean and variance (in fact, finite exponential moments). It is easy to see that, in this case,
the mean of $Y_1$ is precisely $\mu$
\begin{equation}\label{mean}
 \mathbb{E}[Y_1]=\mu,
\end{equation}
and the variance $\sigma^2$ is given by (e.g. \cite{ContTankov04}, proposition 3.13)
\begin{equation}
   \sigma^2= \mbox{Var}(Y_1) = c^2 + \int_{\mathbb{R}} y^2\nu(dy). \label{eq:VarianceofLevy}
    \end{equation}

On the other hand, the random variable $Y_1$ has a symmetric distribution with location $\mu$ and scale $\sigma$.
As such its characteristic function takes the form
\begin{equation}
\varphi_{Y_1}(u)=e^{iu\mu} \psi\left(\frac{\sigma^2}{2}u^2\right),\label{eq:CFSymmetric}
\end{equation}
where the function $\psi(u):[0,\infty)\rightarrow \mathbb{R}$ is
called the characteristic generator of the symmetric family (e.g.
\cite{FangOthers}, p.32). $\psi$ is  unique up to scaling, and if
chosen such that $\psi'(0)=-1$, yields that $\mu$ and $\sigma^2$ are
the mean and variance of $Y_1$ respectively. We denote by $S(\mu,\sigma^2,\psi)$
the distribution whose characteristic function is of the form \eqref{eq:CFSymmetric}.

A detailed account of the properties of symmetric distributions (also known as elliptical distributions)
is given in Fang et al. \cite{FangOthers}.

\subsection{Symmetric L\'evy Processes and Marginals}

The following proposition relates the characteristic triplet of a symmetric L\'evy process
$(Y_t)_{t\geq 0}$ to the parameters of the symmetric distribution of $Y_1$.
\begin{proposition}\label{tripleSym}
Let $(Y_t)_{t\geq 0}$ be a symmetric L\'evy process with characteristic triplet $(\mu,c,\nu)$.
Then $Y_1$ has distribution $S(\mu,\sigma^2,\psi)$ where $\sigma^2$ is given by \eqref{eq:VarianceofLevy},
and $\psi$ by
\begin{equation}
\psi(v) = \exp\left\{-\frac{c^2v}{\sigma^2} - 2\int_0^{\infty}\big(1 - \cos ( y\sqrt{2v}/\sigma) \big) \nu(dy) \right\},\label{eq:SymCharGen}
\end{equation}
and $v = \frac{\sigma^2 u^2}{2}$. Furthermore, $Y_t$ has distribution $S(\mu t,\sigma^2 t,\psi_t)$ with
\begin{equation}
\psi_t(v) = \big(\psi(v/t)\big)^t.\label{eq:gt}
\end{equation}
\end{proposition}

 \begin{proof}
The proof is a straightforward examination of the characteristic function.  The form of $\psi_t$ is
due to
$$\mathbb{E}\left[e^{iuY_t}\right] =  \big(\mathbb{E}\left[e^{iuY_1}\right]\big)^t=\big(\varphi_{Y_1}(u)\big)^t=e^{iu\mu t}
    \left(\psi\left(\frac{\sigma^2t}{2t}u^2\right)\right)^t.$$
\end{proof}

    \subsection{Equivalent Change of Measure for L\'evy Processes}

In general, a L\'evy process under an equivalent measure need not
remain L\'evy, as independence of increments may not be preserved.
However, there is a class of equivalent measures under
which it does.

\begin{theorem}\label{ChangeMeasure}
Let $Y_t$ be a L\'evy process on $\mathbb{R}$ with characteristic triplet $(\mu, c, \nu)$ under $\mathbb{P}$. Let
$\eta \in \mathbb{R}$ and a function $\phi$ be arbitrary such that
    \[
    \int_{\mathbb{R}}\big(e^{\phi(y)/2} - 1\big)^2 \nu(dy) < \infty.
    \]
Then
\begin{enumerate}
\item The limit \[
        \lim_{\epsilon \downarrow 0}\left(\sum_{s \leq t,~|\Delta Y_s| > \epsilon} \phi(\Delta Y_s) - t\int_{|y|>\epsilon}\big(e^{\phi(y)} - 1\big)\nu(dy)\right)
    \]
exists (uniformly in $t$ on any bounded interval).
\item The process
\begin{eqnarray*}
D_t & = & \eta Y_t^c - \frac{\eta^2c^2t}{2} - \eta \mu t\\
& & + \lim_{\epsilon \downarrow 0}\left(\sum_{s \leq t,~|\Delta Y_s| > \epsilon} \phi(\Delta Y_s) - t\int_{|y|>\epsilon}\big(e^{\phi(y)} - 1\big)\nu(dy)\right),
\end{eqnarray*}
where $Y_t^c$ is the continuous part of $Y_t$,
defines a probability measure $\mathbb{Q}$ equivalent to $\mathbb{P}$ by
\begin{equation}\label{equivalentmeasure}
    \frac{d\mathbb{Q}}{d\mathbb{P}}\bigg|_{\mathcal{F}_t} = e^{D_t}.
\end{equation}

\item The process $Y_t$ remains a L\'evy process
under $\mathbb{Q}$ with characteristic triplet $(\tilde{\mu}, c,\tilde{\nu})$, where
    \[
    \tilde{\mu} = \mu + \int_{-1}^{1} y(\tilde{\nu} - \nu)(dy) +
    c^2\eta
    \]
    and
    \[
    \tilde{\nu}(dy) = e^{\phi(y)}\nu(dy).
    \]

\item Conversely, any probability measure equivalent to $\mathbb{P}$ under which $Y_t$ remains a L\'evy process
must be of the form \eqref{equivalentmeasure} and the characteristic triplet must be $(\tilde{\mu}, c,\tilde{\nu})$
as specified above.
\end{enumerate}
\end{theorem}

\begin{proof}
This theorem is a direct consequence of Lemma 33.6 and Theorems 33.1 and 33.2 of Sato \cite{Sato1999}.
Statements 1. and 2. follow from Lemma 33.6. The L\'evy property of the process $Y_t$ under $\mathbb{Q}$ is a consequence of Theorem 33.2.
The form of its characteristic triplet is given by Theorem 33.1. Statement 4. is given by Theorems 33.1 and 33.2.
\end{proof}


Next we consider the case of a symmetric L\'evy process $Y_t$. In order that it
remains symmetric under an equivalent measure $\mathbb{Q}$, we show that it is
necessary and sufficient that the function $\phi$ given in the above theorem be even.
We also describe how the parameters of the symmetric family transform under the change of measure.

\begin{theorem}
\label{PhiCondition}
Let $Y_t$ be a L\'evy process on $\mathbb{R}$ with characteristic triplet $(\mu, c, \nu)$ under $\mathbb{P}$,
and $\mathbb{Q}$ be any ``L\'evy-preserving'' equivalent change of measure as described in Theorem \ref{ChangeMeasure}.
Then $Y_t$ is symmetric, or equivalently the L\'evy measure $\tilde{\nu}$ is symmetric, if and only if   $\phi(-y) =
\phi(y)$ $\nu$-a.e., and in this case, the $\mathbb{Q}$-distribution of $Y_1$ is $S(\tilde{\mu},
\tilde{\sigma}^2, \tilde{\psi})$ where
    \begin{equation}
        \tilde{\mu} = \mu + c^2\eta,
    \end{equation}

    \begin{equation}
        \tilde{\sigma}^2 = c^2 + \int_{\mathbb{R}} y^2e^{\phi(y)}\nu(dy), \label{eq:QSymVariance}
    \end{equation}

    \begin{equation}
        \tilde{\psi} (v) = \exp\left\{-\frac{c^2v}{\tilde{\sigma}^2} - 2\int_{0}^{\infty}\big(1 - \cos (y\sqrt{2v}/\tilde{\sigma})\big)e^{\phi(y)}\nu(dy)\right\}. \label{eq:QSymCharacteristic}
    \end{equation}
\end{theorem}

\begin{proof}
Since $\tilde{\nu}(dy)=e^{\phi(y)}\nu(dy)$, the evenness of $\phi$ is clearly equivalent to the symmetry of $\tilde{\nu}$.

Furthermore, if $\tilde{\nu}$ and $\nu$ are both symmetric, then $\int_{-1}^{1} y(\tilde{\nu} - \nu)(dy) = 0$.
Hence, $\tilde{\mu} = \mu + c^2\eta$. The other two parameters are obtained
from \eqref{eq:VarianceofLevy} and \eqref{eq:SymCharGen}.
\end{proof}

\section{The Natural Change of Measure}

Consider a symmetric L\'evy process $Y_t$ with $\mathbb{P}$-characteristic triplet
$(\mu, c, \nu)$, and $\mathbb{P}$-distribution of $Y_1$,  $S(\mu,\sigma^2, \psi).$

We call an equivalent measure $\mathbb{Q}$ natural for $Y_t$ if $Y_t$ is L\'evy under $\mathbb{Q}$ and
the $\mathbb{Q}$-distribution of $Y_1$ belongs to a family of
symmetric distributions with same characteristic generator, $\tilde{\psi} = \psi$.

When searching for a natural change of measure, an interesting fact emerges; there is, up to a constant,
a unique natural equivalent measure for each L\'evy process. Furthermore, the specific change of measure
takes a dichotomous form depending on whether or not a Brownian component is present.

The next two theorems detail these facts. We start with the uniqueness result, then show existence.

\begin{theorem}
\label{Natural} Let $Y_t$ be a symmetric L\'evy process with
$\mathbb{P}-$characteristic triplet $(\mu, c, \nu)$, and $\mathbb{P}$-distribution of
$Y_1$ $S(\mu, \sigma^2, \psi)$. Suppose $\mathbb{Q}$ is a natural change of
measure.
\begin{enumerate}
\item[(1)] If $c\neq0 $ (a Brownian component is present), then under $\mathbb{Q}$, the
characteristic triplet  becomes $(\tilde{\mu}, c, \nu)$, where
$\tilde{\mu} = \mu + c^2\eta,$ for some $\eta$. In this case, $c$ and $\nu$ remain unchanged.
\item[(2)] If $c = 0 $ (no Brownian component is present), then under $\mathbb{Q}$ the characteristic triplet becomes $(\mu, 0, \tilde{\nu})$ where
$\tilde{\nu}(A) = \int 1_{\{A\}}(\beta y) \nu(dy)$, for some $\beta>0$. In this case, $\mu$ and $c$ remain unchanged.
\end{enumerate}
\end{theorem}

\begin{proof}
Since we only consider equivalent measures $\mathbb{Q}$ that preserve L\'evy
property, we denote the characteristic triplet under $\mathbb{Q}$ with
$(\tilde \mu,c,\tilde\nu)$. By Theorem \ref{ChangeMeasure} we know that
\begin{equation}\label{existsphi}
\tilde{\nu}(dy) = e^{\phi(y)}\nu(dy)
\end{equation}
for some $\phi$.

The proof of the Theorem uses Theorem \ref{ChangeMeasure} and an analytical lemma.

Using expressions \eqref{eq:SymCharGen} and \eqref{eq:QSymCharacteristic},
we can see that $\mathbb{Q}$ is natural (ie $\psi=\tilde \psi$) if and only if for
all $v>0$, the function $\phi$ in \eqref{existsphi} is even and
satisfies the following integral equation
\begin{equation}
\int_0^{\infty} \Big[ \big(1 - \cos(y\sqrt{2v}/\tilde{\sigma})\big)e^{\phi(y)} - \big(1 - \cos(y\sqrt{2v}/\sigma)\big) \Big]
\nu(dy) + \frac{c^2v}{2}\Big(\frac{1}{\tilde{\sigma}^2} - \frac{1}{\sigma^2}\Big) = 0. \label{eq:ConditionG}
\end{equation}
In Lemma \ref{DominatedConvergence} we show that
\begin{equation}\label{lemma1}
\lim_{v\to\infty}\int_0^{\infty} \frac{1}{v}\Big[ \big(1 - \cos(y\sqrt{2v}/\tilde{\sigma})\big)e^{\phi(y)} -
\big(1 - \cos(y\sqrt{2v}/\sigma)\big) \Big] \nu(dy)=0.
\end{equation}

Hence by dividing by $v$ and taking limit in \eqref{eq:ConditionG},
we obtain that
\begin{equation}\label{equalpsi}
\frac{c^2}{2}\Big(\frac{1}{\sigma^2} - \frac{1}{\tilde{\sigma}^2}\Big) = 0,
\end{equation}
which reduces \eqref{eq:ConditionG} to
\begin{equation}
\int_0^{\infty} \Big[ \big(1 - \cos(y\sqrt{2v}/\tilde{\sigma})\big)e^{\phi(y)} - \big(1 - \cos(y\sqrt{2v}/\sigma)\big) \Big]\nu(dy) = 0. \label{eq:ConditionGsimple}
\end{equation}

(1) Consider first the case $c\not=0$.  It follows  from
\eqref{equalpsi} that $\sigma^2=\tilde{\sigma}^2$, and re-parameterizing \eqref{eq:ConditionGsimple} using $\omega=\frac{\sqrt{2v}}{\sigma}=\frac{\sqrt{2v}}{\tilde{\sigma}}$
we get that
$$\int_0^{\infty} \big(1 - \cos(\omega y)\big) \tilde{\nu}(dy) =
\int_0^{\infty} \big(1 - \cos(\omega y)\big) \nu(dy), \hspace{5mm} \forall \omega > 0.$$
It is now, at least intuitively, clear that this implies that $\tilde{\nu} = \nu$.
This is however not straightforward and requires a detailed proof. As it is a technical matter, it is given in the Appendix in Lemma \ref{TransformIdentity}.
Note that in this case, and since $\tilde{\sigma}=\sigma$,
$$\int_0^\infty y^2\tilde{\nu}(dy) = \int_0^\infty y^2\nu(dy).$$

(2) Consider now the case $c=0$. Clearly $\tilde{\mu} = \mu + c^2\eta = \mu$.
Also, with $\beta=\tilde{\sigma}/\sigma$, $\lambda=\sqrt{2v}/\tilde{\sigma}$ and $\nu_{\beta}(dy)=\nu(\frac{1}{\beta}dy)$, \eqref{eq:ConditionGsimple}
becomes
$$\int_0^{\infty} \big(1 - \cos(\lambda y )\big)\tilde{\nu}(dy) =\int_0^{\infty} \big(1 - \cos(\lambda y)\big) \nu_\beta(dy),\;\;\forall \lambda>0.$$
Again, using Lemma \ref{TransformIdentity} we get that $\tilde{\nu} = \nu_{\beta}$ a.e.
\end{proof}

\begin{theorem}\label{suffcond}Let $Y_t$ be a symmetric L\'evy process with
$\mathbb{P}$-characteristic triplet $(\mu,c,\nu)$, and $\mathbb{P}$-distribution of
$Y_1$ $S(\mu,\sigma^2,\psi)$.
\begin{enumerate}
\item[(1)] If $c\neq0$ (a Brownian component is present), then for any $\eta$ there is a
natural change of measure $\mathbb{Q}$, such that the
characteristic triplet becomes $(\tilde{\mu},c,\nu)$, where
$\tilde{\mu} = \mu + c^2\eta$.
In this case, the $\mathbb{Q}$-distribution of $Y_1$ is $ S(\tilde{\mu},\sigma^2,\psi)$.
\item[(2)] If $c=0 $ (no Brownian component is present), then for any $\beta>0$ there is a
natural change of measure $\mathbb{Q}$, such that the characteristic triplet becomes $(\mu,0,\tilde{\nu})$, where
$\tilde{\nu}(A) = \int 1_{\{A\}}(\beta y)\nu(dy)$.
In this case, the $\mathbb{Q}$-distribution of $Y_1$ is $S(\mu,\tilde{\sigma}^2,\psi)$, where $\tilde{\sigma}=\beta\sigma$.

\end{enumerate}
\end{theorem}
\begin{proof}
The proof immediately follows from Theorems \ref{ChangeMeasure} and \ref{PhiCondition}, as well as from the proof of Theorem \ref{Natural}. It is also easy to check that the
$\mathbb{Q}$-distribution of $Y_1$ is as claimed.
\end{proof}


\section{Option Pricing with a Natural EMM}\label{OptionPricing}

\subsection{Natural Equivalent Martingale Measures}

Let now $S_t=S_0e^{Y_t}$ be a model for stock prices, where $Y_t$ is
a symmetric L\'evy process. According to the Fundamental Theorems of
Mathematical Finance, options on stock are priced by using an EMM
$\mathbb{Q}$, under which the discounted stock price process $e^{-rt}S_t,~ 0
< t \leq T$ is a martingale.

To the requirement that $\mathbb{Q}$ be a natural equivalent measure we now add the condition that it also be a martingale measure.

\begin{theorem}\label{Condition1}
\begin{enumerate}
\item[(1)] Let $\mathbb{Q}$ be a natural EMM for a symmetric L\'evy
process, then the following relation must hold between the parameters of the $\mathbb{Q}$-distribution of $Y_1$,
\begin{equation}\label{mainrel}
\tilde \mu+\ln \psi\Big(-\tilde{\sigma}^2/2\Big)=r.
\end{equation}
\item[(2)] If a Brownian component is present ($c\not=0$) and
$\mathbb{Q}$ is a natural EMM, then
\begin{equation}\label{NEMM1}
\tilde{\mu} = r - \ln \psi\left(-\sigma^2/2\right).
\end{equation}
Further, such $\mathbb{Q}$ exists and is unique.
\item[(3)] If   a Brownian component is absent ($c=0$) and $\mathbb{Q}$ is  a natural
EMM, then
    $\tilde{\sigma}^2$ is a root
of equation
    \begin{equation}\label{NEMM2}
    \ln \psi\Big(-\tilde{\sigma}^2/2\Big) =  r - \mu.
    \end{equation}
Further, such $\mathbb{Q}$ exists if and only if the $\mu<r$, and  when it
exists it is unique.
\end{enumerate}
\end{theorem}

\begin{proof}
Imposing the martingale property under $\mathbb{Q}$ to $e^{-rt}S_t$ leads to the requirement that
$$\mathbb{E}_\mathbb{Q}\left[e^{Y_{t-s}}\right] = e^{r(t-s)}.$$
On the other hand, since under the natural EMM $Y_1$ has distribution $S(\tilde\mu,\tilde{\sigma}^2,\psi)$, we have
$$\mathbb{E}_\mathbb{Q}\left[e^{Y_1}\right] =e^{\tilde\mu}\psi\Big(-\tilde{\sigma}^2/2\Big).$$
Therefore the martingale property holds if and only if
$$e^{\tilde\mu}\psi\Big(-\tilde{\sigma}^2/2\Big)=e^r,$$
 which is
equivalent to \eqref{mainrel}.

By the natural change of measure Theorem \ref{Natural},
if $c\not=0$,  only $\mu$ can be changed, consequently we obtain \eqref{NEMM1}.
If $c=0$, only $\sigma^2$ can be changed, hence  we obtain \eqref{NEMM2}.

When $c=0$, the requirement that $\mu<r$ immediately follows by Jensen's inequality:
$\mathbb{E}_\mathbb{Q}\big[e^{Y_1}\big]\geq e^{\mu}$.
\end{proof}

\begin{remark}[Discrete time] In discrete time the natural EMM always
exists, in contradiction with the continuous-time setting of L\'evy process without Brownian
component. Furthermore, the parameters satisfy \eqref{mainrel} -- for
details, see \cite{KlebanerLandsman07}. In discrete time, the natural EMM always exists (and is unique)
and is obtained by changing the location $\mu$ while keeping $\sigma^2$.
\end{remark}

\subsection{Option Pricing with a Natural EMM}
According to the method of pricing by no arbitrage, the value of the option at time $0$ is given by the expectation of the payoff function, i.e.,
\begin{equation}
C_0 = e^{-rT}\mathbb{E}_\mathbb{Q}\big[\big(S_T - K\big)^+\big] \label{eq:OPF}
\end{equation}
where $T$ is the time to maturity, $K$ is the strike price, and $\mathbb{Q}$ is an EMM.
The above formula \eqref{eq:OPF} is arbitrage-free even when $\mathbb{Q}$ is not unique (\cite{EberleinJacob97}, \cite{ShiryaevKruzhilin1999} p.398).

In this section we write the option pricing formula \eqref{eq:OPF} using change of
num\'eraire (see \cite{GemanOthers95}, or \cite{Klebaner2005}, section
11.5), which gives
\begin{equation}
C_0 = S_0\mathbb{Q}_1\big(S_T > K\big) - e^{-rT}K\mathbb{Q}\big(S_T > K\big),\label{eq:OP1}
\end{equation}
where $\mathbb{Q}_1$, defined by
\begin{equation}
\frac{d\mathbb{Q}_1}{d\mathbb{Q}} = e^{-rT}\frac{S_T}{S_0}, \label{eq:DQ1}
\end{equation}
is the measure under which the process $e^{rt}/S_t$ is a martingale.

\subsubsection{Symmetric L\'evy Returns with Brownian Component}\label{S1}

Let $Y_t$ be a symmetric L\'evy process with $\mathbb{P}$-characteristic triplet
$(\mu,c,\nu)$ and such that $c\neq0$. Let $S(\mu,\sigma^2,\psi)$ be the
$\mathbb{P}$-distribution of $Y_1$ and $\mathbb{Q}$ be the natural EMM (for $Y_t$).
Then, under $\mathbb{Q}$, $Y_t$ remains a
symmetric L\'evy process with characteristic triplet $(\tilde{\mu},c,\nu)$, and the
distribution of $Y_1$ becomes $S(\tilde{\mu},\sigma^2,\psi)$, where
$$\tilde{\mu} = r - \ln \psi\big(-\sigma^2/2\big).$$

Now, it is easy to see that $\mathbb{Q}_1$ is also a natural EMM. Indeed, since $-Y_t$ is
a L\'evy process with $\mathbb{P}$-characteristic triplets $(-\mu,c,\nu)$, and since
the distribution of $-Y_1$ is $S(-\mu,\sigma^2,\psi)$, $\mathbb{Q}_1$ is chosen so that
\begin{equation}
\tilde{\mu}_1 = r + \ln\psi\Big(-\frac{\sigma^2}{2}\Big). \label{eq:MStar}
\end{equation}
This choice is unique as the location parameter ($\tilde{\mu}_1$) uniquely determines $\eta$,
which in turn specifies the equivalent measure.
By the uniqueness of $\tilde{\mu}_1$, $\mathbb{Q}_1$ is unique.

Under $\mathbb{Q}_1$, $Y_t$ is a symmetric L\'evy process with marginals from the family $S(\tilde{\mu}_1t,\sigma^2t,\psi_t)$.

\begin{proposition}
Denote by $F_T$ the $\mathbb{P}$-distribution function of the standardized variavle $(Y_T-\mu T)/(\sigma\sqrt{T})$:
$$F_T(y) = \mathbb{P}\big(Y_T\leq\sigma\sqrt{T}y+\mu T\big).$$
Then
$$F_T(y) = \mathbb{Q}\big(Y_T\leq\sigma\sqrt{T}y+\tilde{\mu}T\big) = \mathbb{Q}_1\big(Y_T\leq\sigma\sqrt{T}y+\tilde{\mu}_1T\big),$$
and the option pricing formula \eqref{eq:OP1} becomes
\begin{eqnarray}
C_0 &=& S_0F_T\left(\frac{\ln\big(\frac{S_0}{K}\big) + \big(r + \ln \psi(-\sigma^2/2)\big)T}{\sigma\sqrt{T}} \right) \nonumber \\
& & - e^{-rT}KF_T\left(\frac{\ln\big(\frac{S_0}{K}\big) + \big(r - \ln \psi(-\sigma^2/2)\big)T}{\sigma\sqrt{T}} \right). \label{eq:EF1}
\end{eqnarray}
\end{proposition}
\begin{proof}
The first statement follows from the fact that the distribution of $(Y_T-\mathbb{E}[T])/(\sigma\sqrt{T})$ is the same for all three
measures $\mathbb{P}$, $\mathbb{Q}$ and $\mathbb{Q}_1$; its characteristic function, under all three probabilities, is given
by $\big(\psi\left(u^2/(2T)\right)\big)^T$.

Also, since $Y_T$ is symmetric about $\mu T$, $F_T$ is symmetric about zero and $1 - F_T(a) = F_T(-a)$.
\eqref{eq:EF1} now follows by simple arithmetic.
\end{proof}

\subsubsection{Symmetric L\'evy Returns without Brownian Component}
\label{S2}

Consider now the case when $Y_t$ is a symmetric L\'evy process with $\mathbb{P}$-characteristic triplet
$(\mu,0,\nu)$. Suppose further that the interest rate $r$ is greater than the location parameter $\mu$.
As before, let $S(\mu,\sigma^2,\psi)$ be the
$\mathbb{P}$-distribution of $Y_1$ and $\mathbb{Q}$ be the natural EMM (for $Y_t$).
Then the $\mathbb{Q}$-distribution of $Y_1$ becomes $S(\mu,\tilde{\sigma}^2,\psi)$, where
$\tilde{\sigma}^2$ is the solution of the equation \eqref{NEMM2}.

Unlike the case of symmetric L\'evy returns with a Brownian component
the EMM $\mathbb{Q}_1$ does not define a natural change of measure.
However, in specific cases considered here (Variance Gamma and Normal Inverse Gaussian),
we are able to identify the distributions of $(Y_T-\mu T)/(\tilde{\sigma}\sqrt{T})$ under $\mathbb{Q}$ and
that of $(Y_T-\mu T)/(\tilde{\sigma}_1\sqrt{T})$ under $\mathbb{Q}_1$. Denoting by $F_T$
and $F_T^1$ the respective cumulative distribution functions, we can write
\begin{eqnarray}
C_0 &=& S_0\mathbb{Q}_1(S_T > K) - e^{-rT}K\mathbb{Q}(S_T > K) \label{eq:EF2} \\
&=& S_0\left[1 - F_T^1\left(-\frac{\ln\big(\frac{S_0}{K}\big) + \mu T}{\tilde{\sigma}_1\sqrt{T}}\right)\right]
         - e^{-rT}KF_T\left(\frac{\ln\big(\frac{S_0}{K}\big) + \mu T}{\tilde{\sigma}\sqrt{T}} \right). \nonumber
\end{eqnarray}

\section{Variance Gamma Model}\label{LSVG}
Here the stock price is modelled as $S_t=S_0e^{Y_t}$, where $Y_t$ is a Variance-Gamme (VG) process.
The marginal distributions of the VG process was originally given in
\cite{MadanCarrChang98} in terms of special functions
 involving the
modified Bessel function of the second kind and the degenerate
hypergeometric function. In the special case of symmetric processes
the marginals turn out to be symmetric Bessel distributions. This
observation leads to elegant formulae.

\subsection{Symmetric Variance Gamma Process } \label{SVGP}

Denote by $Bessel(\mu, \sigma^2, \lambda)$ the Bessel distribution with mean
 $\mu$, variance $\sigma^2$ and shape parameter $\lambda$. A   symmetric Bessel distribution has mean $\mu = 0$, and   characteristic function (\cite{JohnsonKotzBala94}, p.51) of the form
    \[
    \varphi(u) = \bigg(\frac{1}{1 + \frac{u^2 \sigma^2}{2\lambda}} \bigg)^{\lambda}.
    \]
The density function of the symmetric Bessel distribution is given by (\cite{JohnsonKotzBala94}, p.50)
    \begin{equation}
    f (x) = \sqrt{\frac{2\lambda}{\pi \sigma^2}}\bigg(\sqrt{\frac{\lambda x^2}{2\sigma^2}}\bigg)^{\lambda - \frac{1}{2}}\frac{1}{\Gamma(\lambda)} K_{\lambda - \frac{1}{2}}\bigg(2\sqrt{\frac{\lambda x^2}{2\sigma^2}}\bigg), \label{eq:SymBesselDensity2}
    \end{equation}
where $K_w(.)$ is the modified Bessel function of the second kind.
     We always consider symmetric Bessel distribution shifted by   $\mu$ but we will drop the word ``shifted''.
     The symmetric Bessel distribution $Bessel(\mu , \sigma^2, \lambda)$ belongs to the family of symmetric
      distributions $S(\mu, \sigma^2, \psi)$ with characteristic
      generator
    \begin{equation}
    \psi(v) =  \bigg(\frac{1}{1 + \frac{v}{\lambda}} \bigg)^{\lambda}. \label{eq:CharGenBessel}
    \end{equation}
We note that the kurtosis of symmetric Bessel distribution is $3 +
\frac{3}{\lambda}$ and hence, the shape parameter $\lambda$ is
related to the excess kurtosis by $\lambda = \frac{3}{\gamma}$.
(since the excess kurtosis of a random variable $Y$ is $\gamma =
\frac{\mathbb{E}[(Y - \mu)^4]}{\sigma^4} - 3$).

A symmetric VG  $Y_t$  has characteristic function
    \begin{equation}
    \mathbb{E}\left[e^{iuY_t}\right] = e^{iu\mu t}\bigg(\frac{1}{1 + \frac{u^2 \sigma^2\kappa}{2}} \bigg)^{\frac{t}{\kappa}} = e^{iu\mu t}\bigg(\frac{1}{1 + \frac{u^2 \sigma^2t}{2\lambda t}} \bigg)^{\lambda t}, \label{eq:CharSymVGt2}
    \end{equation}
in which we have employed $\lambda = \frac{1}{\kappa}$.  By
inspecting the characteristic function we have

\begin{proposition}\label{VGBessel}
The marginals of a symmetric Variance Gamma process
$Y_t$ is a symmetric Bessel distribution with mean $\mu t$,
variance $\sigma^2 t$ and shape parameter $\lambda t$, i.e., $Y_t \sim Bessel(\mu t, \sigma^2t, \lambda t)$, which belongs to the family of symmetric distributions $S(\mu t, \sigma^2t, \psi_t)$ where the characteristic generator is given by
    \begin{equation}
    \psi_t(v) = \big(\psi(v/t)\big)^t = \bigg(\frac{1}{1 + \frac{v}{\lambda t}} \bigg)^{\lambda t}. \label{eq:CharGenVG}
    \end{equation}
\end{proposition}

\subsection{Option Pricing with Symmetric VG}
\subsubsection{Continuous Time}
We determine the distributions of $Y_t$ under the EMMs $\mathbb{Q}$ and
$\mathbb{Q}_1$. By Theorems \ref{Natural} and \ref{Condition1}, if  $\mu <
r$, the $\mathbb{Q}$-distribution of $Y_1$ is symmetric Bessel
$Bessel(\mu, \tilde{\sigma}^2, \psi)$ where by using
\eqref{eq:CharGenBessel} we get from \eqref{NEMM2}
    \begin{align}
    \tilde{\sigma}^2 = 2\lambda(1 - e^{-(r - \mu)/\lambda}). \label{eq:QVarVG}
    \end{align}

Under $\mathbb{Q}_1$, the distribution of $Y_t$ is identified by the following.

\begin{proposition}\label{Q1Density}
 Denote by $f_{Y_t}^{\mathbb{Q}}$ the $\mathbb{Q}$-density of $Y_t \sim Bessel(\mu t, \tilde{\sigma}^2t, \lambda t)$.
 Then the density of $Y_t$ under $\mathbb{Q}_1$,   given by $e^{y - rt}f_{Y_t}^{\mathbb{Q}}(y)$, is the density function of an asymmetric Bessel distribution.
\end{proposition}

\begin{proof}
From the change of num\'eraire defined by \eqref{eq:DQ1}, it can be seen that the density of $Y_t$ under $\mathbb{Q}_1$ is given by
    \begin{align}
    &e^{y - rt} f_{Y_t}^{\mathbb{Q}}(y) \nonumber\\
    &= e^{y - rt} \sqrt{\frac{2\lambda}{\pi \tilde{\sigma}^2}}\bigg(\sqrt{\frac{\lambda (y - \mu t)^2}{2\tilde{\sigma}^2}}\bigg)^{\lambda t - \frac{1}{2}}\frac{1}{\Gamma(\lambda t)} K_{\lambda t - \frac{1}{2}}\bigg(2\sqrt{\frac{\lambda (y - \mu t)^2}{2\tilde{\sigma}^2}}\bigg). \label{eq:Den1}
    \end{align}

Using the characteristic generator \eqref{eq:CharGenVG} for the symmetric Bessel distribution, it follows that
    \begin{equation}
    e^{y - rt} = e^{y - \mu t}\bigg(1 - \frac{\tilde{\sigma}^2}{2\lambda}\bigg)^{\lambda t}. \label{eq:EYRT}
    \end{equation}
Now, apply \eqref{eq:EYRT} and let $y^* = y - \mu t$, the expression in the second line of (\ref{eq:Den1}) becomes
    \begin{align*}
    e^{y^*}\bigg(1 - \frac{\tilde{\sigma}^2}{2\lambda}\bigg)^{\lambda t} \sqrt{\frac{2\lambda}{\pi \tilde{\sigma}^2}}\bigg(\sqrt{\frac{\lambda (y^*)^2}{2\tilde{\sigma}^2}}\bigg)^{\lambda t - \frac{1}{2}}\frac{1}{\Gamma(\lambda t)} K_{\lambda t - \frac{1}{2}}\bigg(\sqrt{\frac{2\lambda (y^*)^2}{\tilde{\sigma}^2}}\bigg).
    \end{align*}
Subsequently, let $m = \lambda t - \frac{1}{2}$, $a = -\sqrt{\frac{\tilde{\sigma}^2}{2\lambda}}$ and $b = \sqrt{\frac{\tilde{\sigma}^2}{2\lambda}}$, we obtain
    \begin{align}
    e^{y - rt} f_{Y_t}(y)
    &= e^{y^*}\big(1 - a^2\big)^{m + \frac{1}{2}} \frac{1}{b\sqrt{\pi}}\bigg(\frac{|y^*|}{2b}\bigg)^{m}\frac{1}{\Gamma(m + \frac{1}{2})} K_{m}\bigg(\frac{|y^*|}{b}\bigg) \nonumber \\
    &= \frac{(1 - a^2)^{m + \frac{1}{2}}|y^*|^m}{\sqrt{\pi}2^m b^{m + 1}\Gamma(m + \frac{1}{2})} e^{y^*} K_{m}\bigg(\bigg|\frac{y^*}{b}\bigg|\bigg). \label{eq:Den2}
    \end{align}
On closer observation, we immediately recognize that the density written in the form \eqref{eq:Den2} is the density of an asymmetric Bessel function distribution (\cite{JohnsonKotzBala94}, p.50), which has the form
    \begin{equation}
        f_Z(z) = \frac{|1 - a^2|^{m + \frac{1}{2}}|z|^m}{\sqrt{\pi}2^m b^{m + 1}\Gamma\big(m + \frac{1}{2}\big)}e^{-\frac{az}{b}} K_m\Big(\Big|\frac{z}{b}\Big|\Big). \label{eq:BesselDensity2}
    \end{equation}

\end{proof}

For completeness we give explicit expressions of the mean, variance,
skewness and kurtosis of an asymmetric Bessel  distribution
  (\cite{JohnsonKotzBala94}, p.51).
    \begin{align}
    \mbox{Mean} &= (2m + 1)ba(a^2 - 1)^{-1} \label{eq:MeanBessel} \\
    \mbox{Variance} &= (2m + 1)b^2(a^2 + 1)(a^2 - 1)^{-2} \label{eq:VarianceBessel}\\
    \mbox{Skewness} &= 2a(a^2 + 3)(2m + 1)^{-1/2}(a^2 + 1)^{-3/2} \label{eq:SkewnessBessel}\\
    \mbox{Kurtosis} &= 3 + 6(a^4 + 6a^2 + 1)(2m + 1)^{-1}(a^2 + 1)^{-2} \label{eq:KurtosisBessel}
    \end{align}
where $a = -\sqrt{\frac{\tilde{\sigma}^2}{2\lambda}}$, $b = -a$, and $m = \lambda t - \frac{1}{2}$.

Let   $Y_t^* = Y_t - \mu t$, it follows from \eqref{eq:Den2} that
$Y_t^*$ has an asymmetric Bessel distribution, denoted
$Bessel^1(\mu_1 t, \tilde{\sigma}_1^2t, \lambda t)$, where the mean
$\mu_1t$ and variance $\tilde{\sigma}_1^2t$ are given by
(\ref{eq:MeanBessel}) and (\ref{eq:VarianceBessel}) respectively. In
particular,
    \begin{align}
    \mathbb{E}[Y_1^*] &= \mu_1 = 2\lambda\big(e^{(r - \mu)/\lambda} - 1\big),\\
    Var(Y_1^*) &= \tilde{\sigma}_1^2 = 2\lambda\big(e^{(r - \mu)/\lambda} - 1\big)\big(2e^{(r - \mu)/\lambda} - 1\big),
    \end{align}
in which we have employed \eqref{eq:QVarVG}. Therefore under $\mathbb{Q}_1$,
$Y_t = Y_t^* + \mu t \sim Bessel^1(\mu t + \mu_1 t,
\tilde{\sigma}_1^2t, \lambda t)$.

Finally, denote by $B_{\lambda t}(y)$ the cumulative distribution function of the standardized symmetric Bessel random variable
 $\frac{Y_t - \mu t}{\tilde{\sigma}\sqrt{t}} \sim Bessel(0, 1, \lambda t)$ under $\mathbb{Q}$, and by $B_{\lambda t}^1(y)$
 the cumulative distribution function of the standardized asymmetric Bessel random
 variable $\frac{Y_t - \mu t - \mu_1 t}{\tilde{\sigma}_1\sqrt{t}} \sim Bessel^1(0, 1, \lambda t)$ under
 $\mathbb{Q}_1$.
\begin{proposition} Let $Y_t \sim Bessel(\mu t, \sigma^2t, \lambda
t)$, and $\mu<r$. Then the arbitrage-free price of a call option
using natural EMM is given by
    \begin{align}
    C_0 &= S_0\left[1 - B_{\lambda T}^1\left(-\frac{\ln\big(\frac{S_0}{K}\big) + \mu T + \mu_1T}{\tilde{\sigma}_1 \sqrt{T}} \right)\right] \nonumber\\
    &\hspace{40mm} - e^{-rT}KB_{\lambda T}\left(\frac{\ln\big(\frac{S_0}{K}\big) + \mu T}{\tilde{\sigma}\sqrt{T}} \right), \label{eq:ExactFormula}
    \end{align}
where $\mu_1 = 2\lambda\big(e^{(r - \mu)/\lambda} - 1\big)$, $\tilde{\sigma}_1^2 = 2\lambda\big(e^{(r - \mu)/\lambda} - 1\big)\big(2e^{(r - \mu)/\lambda} - 1\big)$ and $\tilde{\sigma}^2 = 2\lambda(1 - e^{-(r - \mu)/\lambda})$.

\end{proposition}

\begin{remark}
An option pricing formula for a general VG process was given in
\cite{MadanMilne91} and \cite{MadanCarrChang98} (eq.25).
  While \cite{MadanMilne91} presented the formula as a double integral of elementary functions
  and obtained the price by numerical integration, \cite{MadanCarrChang98} provided a closed form formula in
  terms of the special functions   involving the modified Bessel function of the second kind
   and the degenerate hypergeometric function.  For
   the symmetric case the formula is much simpler.
\end{remark}
\begin{remark}\label{notgood}
The shortcoming of the natural EMM approach in continuous time is
that $\mu<r$. This is overcome by using discrete time, where natural
EMM exists also when $\mu\ge r$, and is given in the next section.
\end{remark}

\subsubsection{Discrete Time}

Let now $S_N = S_0 e^{Y_N}$ be a model for stock prices, where $Y_N
= \sum_{n = 1}^{N} \Delta Y_n$ is a symmetric VG process in discrete
time.   $\Delta Y_n$, $n = 1, \ldots, N$ are i.i.d. symmetric Bessel
distribution $Bessel(\mu, \sigma^2, \lambda)$, which belongs to the
symmetric family $S(\mu, \sigma^2, \psi)$, where $\psi$ is given in
(\ref{eq:CharGenBessel}).

It is possible to chose  both $\mathbb{Q}$ and $\mathbb{Q}_1$ as natural EMM's that
shift only the location parameter. Hence, the $\mathbb{Q}$-distribution of
$\Delta Y_n$ is $Bessel(\tilde{\mu}, \sigma^2, \lambda)$ with
$\tilde{\mu} = r - \ln\psi\big(-\frac{\sigma^2}{2}\big)$. The
$\mathbb{Q}_1$-distribution of $\Delta Y_n$ is $Bessel(\tilde{\mu}_1,
\sigma^2, \lambda)$ with $\tilde{\mu}_1 = r +
\ln\psi\big(-\frac{\sigma^2}{2}\big)$. This is easily seen, see also
\cite{KlebanerLandsman07}. Denote by $B_{\lambda N}(y)$ the
cumulative distribution function of the standardized symmetric
Bessel random variable $\frac{Y_N - \mu N}{\sigma \sqrt{N}}$.

\begin{proposition} Let   $\Delta Y_n$ follow a symmetric Bessel distribution $Bessel(\mu, \sigma^2, \lambda)$, then the arbitrage-free price of a call option with $N$ periods to expiration is given by
    \begin{align}
    C_0 &= S_0B_{\lambda N}\left(\frac{\ln\big(\frac{S_0}{K}\big) + \big(r - \lambda \ln(1 - \frac{\sigma^2}{2\lambda})\big)N}{\sigma \sqrt{N}} \right) \nonumber\\
     & \hspace{10mm} - e^{-rN}KB_{\lambda N}\left(\frac{\ln\big(\frac{S_0}{K}\big) + \big(r + \lambda \ln(1 - \frac{\sigma^2}{2\lambda})\big)N}{\sigma \sqrt{N}} \right). \label{eq:ExactOptionDVG}
    \end{align}
\end{proposition}

\subsubsection{Numerical Comparisons}

For comparisons, we approximate the distributions of the standardized Bessel random variables $(Y_T - \mu T)/\sigma \sqrt{T}$ (the time $T$ is replaced by $N$ in the discrete case) by the standard Normal, in other words, $B_{\lambda T}$ by $\Phi$. We also approximate the standardized asymmetric Bessel random variable that arises in the continuous time case by the standard Normal, because its distribution is only slightly negatively skewed and therefore it is negligible. We will assume this is the case (skewness is small) in our approximation. Moreover, recall that the shape parameter $\lambda$ and the excess kurtosis $\gamma$ of the symmetric Bessel distribution are related by $\lambda = \frac{3}{\gamma}$. Thus, for each of the continuous time and discrete time cases, we obtain an easy to use Black-Scholes type formula for option pricing which gives correction that accounts for the access kurtosis.

In the continuous time case, the generalized or modified Black-Scholes formula for log-symmetric VG model (VG-C) is given by
    \begin{align}
    C_0 \approx S_0\Phi\left(\frac{\ln\big(\frac{S_0}{K}\big) + \mu T + \mu_1 T}{\tilde{\sigma}_1 \sqrt{T}} \right)
    - e^{-rT}K\Phi\left(\frac{\ln\big(\frac{S_0}{K}\big) + \mu T}{\tilde{\sigma}\sqrt{T}} \right), \label{eq:VGC}
    \end{align}
where $\mu_1 = \frac{6}{\gamma}\big(e^{(r - \mu)\gamma/3} - 1\big)$, $\tilde{\sigma}_1^2 = \frac{6}{\gamma}\big(e^{(r - \mu)\gamma/3} - 1\big)\big(2e^{(r - \mu)\gamma/3} - 1\big)$ and $\tilde{\sigma}^2 = \frac{6}{\gamma}(1 - e^{-(r - \mu)\gamma/3})$.
Note that the Black-Scholes formula is a special case of the generalized version (VG-C) (\ref{eq:VGC}) when $\gamma \rightarrow 0$ due to the followings:
    \begin{align*}
    \mu_1 &= \frac{6}{\gamma}\big(e^{(r - \mu)\gamma/3} - 1\big) \rightarrow 2(r - \mu),\\
    \tilde{\sigma}_1^2 &= \frac{6}{\gamma}\big(e^{(r - \mu)\gamma/3} - 1\big)\big(2e^{(r - \mu)\gamma/3} - 1\big) \rightarrow 2(r - \mu),\\
    \tilde{\sigma}^2 &= \frac{6}{\gamma}\big(1 - e^{-(r - \mu)\gamma/3} \big) \rightarrow 2(r - \mu).
    \end{align*}
And if $2(r - \mu) = \sigma^2$, which is a constant as in the Black-Scholes model (Recall that under the risk-neutral measure $\mathbb{Q}$, the mean $\mu = r - \frac{\sigma^2}{2}$ and the volatility $\sigma$ is a constant), then by using these results and some simple manipulations, it is easy to see that the generalized formula (VG-C) (\ref{eq:VGC}) is the exact Black-Scholes formula.

In the discrete time case, the modified Black-Scholes formula for log-symmetric VG model (VG-D) is given by
    \begin{align}
    C_0 &\approx S_0\Phi\left(\frac{\ln\big(\frac{S_0}{K}\big) + \big(r - \frac{3}{\gamma} \ln(1 - \frac{\gamma\sigma^2}{6})\big)N}{\sigma \sqrt{N}} \right) \nonumber\\
     & \hspace{10mm} - e^{-rN}K\Phi\left(\frac{\ln\big(\frac{S_0}{K}\big) + \big(r + \frac{3}{\gamma} \ln(1 - \frac{\gamma\sigma^2}{6})\big)N}{\sigma \sqrt{N}} \right). \label{eq:VGD}
    \end{align}
It can be seen that the Black-Scholes formula is a limit of the generalized version (VG-D) (\ref{eq:VGD}) for every $N$ when $\gamma \rightarrow 0$ due to
    \[
    \frac{3}{\gamma} \ln\Big(1 - \frac{\gamma\sigma^2}{6}\Big) \rightarrow -\frac{\sigma^2}{2}.
    \]

The classical Black-Scholes formula (BS) is considered robust in the sense that for small values of excess kurtosis $\gamma$, it coincides with the modified Black-Scholes formulae in both continuous time and discrete time cases. However, even for the moderate values of $\gamma$, the distinction between the modified Black-Scholes formulae (VG-C and VG-D) and BS is noticeable (see Figure \ref{Fig1}), with the disagreement between VG-C and BS formulae being greater than the disagreement between VG-D and BS. The exact prices and percentage differences are represented in Table \ref{tab:tab1}.

\begin{figure}[!htb]
    \centering
    \includegraphics [bb=60 80 720 510, scale=0.5]{graph1_VG.eps}
    \caption{Option prices and percentage differences obtained by VG-C, VG-D and BS formulae for log-Bessel distribution weekly returns, $S_0 = K = 10$, $r = 0.06$, $\sigma = 0.19$, $\mu = 0.03$, $\gamma = 4$}
    \label{Fig1}
\end{figure}

\begin{table}[!htb]
\centering
\begin{tabular}{|l|c|c|c|c|c|c|}
  \hline

  \begin{minipage}{3cm}
  Time to maturity (weeks)
  \end{minipage}
  & 2 & 12 & 22 & 32 & 42 & 52 \\
  \hline \hline
  BS formula & 0.160 & 0.434 & 0.622 & 0.782 & 0.927 & 1.062 \\
  \hline
  VG-D formula & 0.162 & 0.439 & 0.628 & 0.789 & 0.935 & 1.071 \\
  Percentage difference & 1.13 & 1.01 & 0.94 & 0.88 & 0.84 & 0.80 \\
  \hline
  VG-C formula & 0.192 & 0.511 & 0.725 & 0.904 & 1.065 & 1.213 \\
  Percentage difference & 19.85 & 17.67 & 16.46 & 15.55 & 14.81 & 14.17 \\
  \hline
\end{tabular}
\caption{Option prices and percentage differences obtained by VG-C, VG-D and BS formulae for log-Bessel distribution weekly returns, $S_0 = K = 10$, $r = 0.06$, $\sigma = 0.19$, $\mu = 0.03$, $\gamma = 4$}
\label{tab:tab1}
\end{table}

\begin{remark}
We suggest to use the modified formulae for any distribution with a
positive excess kurtosis.
\end{remark}


\section{Normal Inverse Gaussian Model}

\subsection{The NIG Process and Distribution}
\label{SNIGP}

In this section, the model is $S_t=S_0e^{Y_t}$, where $Y_t$ is a Normal inverse
Gaussian (NIG) process. The NIG distributions were first introduced
by Barndorff-Nielsen \cite{BarndorffNielsen95} as a subclass of the
Generalized Hyperbolic distribution with parameter $\lambda =
-\frac{1}{2}$. Denote by $NIG(\alpha, \beta, \delta, \mu)$ the NIG
distribution, where $\mu$ is the location parameter, $\delta$ is the
scale parameter, $\alpha$ is the shape parameter and $\beta$ is for
skewness. The density of NIG distribution is given by (see e.g.
\cite{Schoutens03} p.60, or \cite{Rydberg97})
    \begin{equation}
    f_Y(y) = \frac{\alpha}{\pi}e^{\delta\sqrt{\alpha^2 - \beta^2} + \beta(y - \mu)}\frac{K_1\Big(\alpha\delta\sqrt{1 + (\frac{y - \mu}{\delta})^2}\Big)}{\sqrt{1 + (\frac{y - \mu}{\delta})^2}}, \label{eq:NIGDensity}
    \end{equation}
where $K_1$ is the modified Bessel function of the third kind, $y,\mu \in \mathbb{R}$, $\delta \geq 0$ and $0 \leq |\beta| \leq \alpha$. The characteristic function of NIG distribution is
    \begin{equation}
    \varphi (u) = e^{iu\mu}\frac{e^{\delta\sqrt{\alpha^2 - \beta^2}}}{e^{\delta\sqrt{\alpha^2 - (\beta + iu)^2}}}, \label{eq:NIG_CF}
    \end{equation}
and the mean, variance, skewness and kurtosis of NIG distribution are
    \begin{align}
    \mbox{Mean} &= \mu + \frac{\beta \delta}{\sqrt{\alpha^2 - \beta^2}} \label{eq:NIGMean}\\
    \mbox{Variance} &= \frac{\delta\alpha^2}{\big(\sqrt{\alpha^2 - \beta^2}\big)^3}. \label{eq:NIGVariance}\\
    \mbox{Skewness} &= \frac{3\beta}{\alpha\big(\delta\sqrt{\alpha^2 - \beta^2}\big)^{\frac{1}{2}}} \label{eq:NIGSkew}\\
    \mbox{Kurtosis} &= 3\bigg(1 + \frac{\alpha^2 + 4\beta^2}{\delta\alpha^2\sqrt{\alpha^2 - \beta^2}}\bigg). \label{eq:NIGKurtosis}
    \end{align}

The symmetric NIG L\'evy processes have symmetric NIG marginals. The NIG distribution is symmetric when the skewness parameter $\beta = 0$. In this case, the density \eqref{eq:NIGDensity} of a symmetric NIG is
    \begin{align}
    f_Y(y) = \frac{\alpha}{\pi}e^{\alpha\delta} \frac{K_1\Big(\alpha\delta\sqrt{1 + (\frac{y - \mu}{\delta})^2}\Big)}{\sqrt{1 + (\frac{y - \mu}{\delta})^2}}.     \label{eq:SNIG_PDF}
    \end{align}
The characteristic function \eqref{eq:NIG_CF} for symmetric NIG is
    \begin{align}
    \varphi (u) = e^{iu\mu}e^{\alpha\delta\big(1 - \sqrt{1 + (\frac{u}{\alpha})^2}\big)}.
    \end{align}
It follows from equations \eqref{eq:NIGMean}, \eqref{eq:NIGVariance} and \eqref{eq:NIGKurtosis}, respectively, that $\mu$ is the mean, variance is $\frac{\delta}{\alpha}$ and kurtosis is $3 + \frac{3}{\alpha\delta}$. We will denote the distribution of symmetric NIG by $SNIG(\alpha, 0, \delta, \mu)$.

 By an inspection of the characteristic function we can see that the characteristic generator of symmetric Normal
inverse Gaussian distributions is given by
    \begin{equation}\label{SNIG_CG}
    \psi(v) = e^{\zeta\big(1 - \sqrt{1 + \frac{2v}{\zeta}}\big)},
    \end{equation}
where $\zeta = \alpha\delta$.

\subsection{Option Pricing with Symmetric NIG Model}
\subsubsection{Continuous Time}

We   determine the distributions of $Y_t$ under the EMM's $\mathbb{Q}$ and
$\mathbb{Q}_1$.

If $\mu < r$ there is a natural EMM and $Y_1$ is symmetric NIG
distribution $SNIG(\alpha, 0, \tilde{\delta}, \mu)$, where
$\frac{\tilde{\delta}}{\alpha} = \tilde{\sigma}^2$. where by using
\eqref{SNIG_CG} we get from \eqref{NEMM2}
    \begin{align}
    \tilde{\sigma}^2 = 2(r - \mu) - \left(\frac{r - \mu}{\alpha\sigma}\right)^2, \hspace{10mm} . \label{eq:NewVar}
    \end{align}
Consequently, $Y_t \sim SNIG(\alpha, 0, \tilde{\delta}t, \mu t)$ under $\mathbb{Q}$ with mean $\mu t$ and variance $\tilde{\sigma}^2t = \tilde{\delta}t/\alpha$.

Under $\mathbb{Q}_1$, the distribution of $Y_t$ is identified in the following theorem.

\begin{proposition}\label{DensityZ1}
The density of $Y_t$ under $\mathbb{Q}_1$, given by $e^{y -
rt}f_{Y_t}^{\mathbb{Q}}(y)$, is the density function of an asymmetric NIG
distribution, $Y_t \sim NIG(\alpha, 1, \tilde{\delta}t, \mu t)$.
\end{proposition}

\begin{proof}
Using \eqref{NEMM2} and \eqref{eq:DQ1}
 \begin{align}
    e^{y - rt} f_{Y_t}^{\mathbb{Q}}(y)
    &= e^{y - \mu t - \alpha\tilde{\delta} t + \tilde{\delta}t\sqrt{\alpha^2 - 1}}
    \frac{\alpha}{\pi}e^{\alpha\tilde{\delta}t} \frac{K_1\Big(\alpha\tilde{\delta}t\sqrt{1 + \big(\frac{y - \mu t}{\tilde{\delta}t}\big)^2}\Big)}{\sqrt{1 + \big(\frac{y - \mu t}{\tilde{\delta}t}\big)^2}} \nonumber\\
    &= \frac{\alpha}{\pi}e^{y-\mu t + \tilde{\delta}t \sqrt{\alpha^2 - 1}} \frac{K_1\Big(\alpha\tilde{\delta}t\sqrt{1 + \big(\frac{y - \mu t}{\tilde{\delta}t}\big)^2}\Big)}{\sqrt{1 + \big(\frac{y - \mu t}{\tilde{\delta}t}\big)^2}}. \label{eq:ADen1}
    \end{align}
One can verify that \eqref{eq:ADen1} is the density function of an asymmetric NIG distribution (see \eqref{eq:NIGDensity}) with parameters $\alpha$ (unchange), $\beta = 1$, $\mu = \mu t$ and $\delta = \tilde{\delta}t.$

\end{proof}

The mean and variance are given by \eqref{eq:NIGMean} and
\eqref{eq:NIGVariance}, respectively.
    \begin{align}
    \mathbb{E}[Y_1] &= \mu_1 = \mu + \sqrt{\frac{\alpha^2}{\alpha^2 - 1}}\tilde{\sigma}^2, \label{eq:Mu1}\\
    Var(Y_1) &= \tilde{\sigma}_1^2 = \left(\sqrt{\frac{\alpha^2}{\alpha^2 - 1}}\right)^3\tilde{\sigma}^2, \label{eq:S1}
    \end{align}
where $\tilde{\sigma}^2$ is given by \eqref{eq:NewVar}.

Denote by $F_{SNIG}(y)$, the cumulative distribution function of the standardized symmetric NIG random variable $\frac{Y_t - \mu t}{\tilde{\sigma}\sqrt{t}} \sim SNIG(\alpha, 0, \alpha, 0)$  under $\mathbb{Q}$, and denote by $F_{NIG}(y)$ the cumulative distribution function of the standardized asymmetric NIG random variable $\frac{Y_t - \mu_1 t}{\tilde{\sigma}_1\sqrt{t}} \sim NIG(\alpha, 1, \alpha, 0)$ under $\mathbb{Q}_1$, we obtain the explicit formula for option pricing with symmetric NIG process that can be written in terms of the cumulative distribution functions of standardized (symmetric and asymmetric) NIG.

\begin{proposition}
Let   $Y_t \sim SNIG(\alpha, 0, \delta t, \mu t)$, and $\mu<r$. Then
the arbitrage-free price of a call option using natural EMM is given
by
    \begin{align}
    C_0 &= S_0\left[1 - F_{NIG}\left(-\frac{\ln\big(\frac{S_0}{K}\big) + \mu_1 T}{\tilde{\sigma}_1\sqrt{T}} \right)\right] \nonumber\\
    &\hspace{30mm} - e^{-rT}KF_{SNIG}\left(\frac{\ln\big(\frac{S_0}{K}\big) + \mu T}{\tilde{\sigma}\sqrt{T}} \right). \label{eq:Exact_NIG}
    \end{align}
where $\mu_1 = \mu + \sqrt{\frac{\alpha^2}{\alpha^2 - 1}}\tilde{\sigma}^2$, $\tilde{\sigma}_1^2 = \left(\sqrt{\frac{\alpha^2}{\alpha^2 - 1}}\right)^3\tilde{\sigma}^2$ and $\tilde{\sigma}^2 = 2(r - \mu) - \left(\frac{r - \mu}{\alpha\sigma}\right)^2$.

\end{proposition}
\noindent Discrete time model allows for natural EMM even when
$\mu\ge r$, Remark \ref{notgood}.

\subsubsection{Discrete Time}
The stock price process in discrete time $S_N = S_0e^{Y_N}$
where $Y_N = \sum_{n = 1}^N \Delta Y_n$ with $\Delta Y_n$, $n = 1,
\ldots N$ are i.i.d.   $SNIG(\alpha, 0, \delta, \mu)$, which belongs
to the symmetric family $S(\mu, \sigma^2, \psi)$ where $\sigma^2 =
\frac{\delta}{\alpha} = \frac{\delta^2}{\alpha\delta}$, and $\psi$
is   \eqref{SNIG_CG}.

Recall that for a fixed $\sigma^2$, we can obtain two natural EMM's
$\mathbb{Q}$ and $\mathbb{Q}_1$ by changing only the location parameter $\mu$ so that
$Y_N$ remains a symmetric NIG L\'evy process. The $\mathbb{Q}$-distribution of
$\Delta Y_n$ is $SNIG(\alpha, 0, \delta, \tilde{\mu})$ where
$\tilde{\mu} = r - \ln\psi\big(-\frac{\sigma^2}{2}\big)$, and the
$\mathbb{Q}_1$-distribution of $\Delta Y_n$ is $SNIG(\alpha, 0,\delta,
\tilde{\mu}_1)$ with $\tilde{\mu}_1 = r +
\ln\psi\big(-\frac{\sigma^2}{2}\big)$. By \eqref{SNIG_CG} and $\zeta
= \alpha\delta = \alpha^2\sigma^2$, we obtain
    \begin{align}
    \ln\psi\Big(-\frac{\sigma^2}{2}\Big) = \zeta \left(1 - \sqrt{1 - \frac{\sigma^2}{\zeta}}\right) = \alpha^2\sigma^2 - \alpha\sigma^2\sqrt{\alpha^2 - 1}. \label{eq:LnPsi2}
    \end{align}
Therefore, we obtain the following result for the exact option pricing formula with symmetric NIG process in discrete time.

\begin{proposition} Let   $\Delta Y_n$ follow a symmetric NIG distribution $SNIG(\alpha, 0, \delta, \mu)$, then the arbitrage-free price of a call option with $N$ periods to expiration is given by
    \begin{align}
    C_0 &= S_0F_{SNIG}\left(\frac{\ln\big(\frac{S_0}{K}\big) + \big(r + \alpha^2\sigma^2 - \alpha\sigma^2\sqrt{\alpha^2 - 1}\big)N}{\sigma \sqrt{N}} \right) \nonumber\\
     & \hspace{10mm} - e^{-rN}KF_{SNIG}\left(\frac{\ln\big(\frac{S_0}{K}\big) + \big(r - \alpha^2\sigma^2 + \alpha\sigma^2\sqrt{\alpha^2 - 1}\big)N}{\sigma \sqrt{N}} \right). \label{eq:ExactOptionDNIG}
    \end{align}
\end{proposition}

\subsection{Numerical Comparisons}

For comparisons, we approximate the standardized symmetric NIG distribution by the standard Normal, in other words, $F_{SNIG}$ by $\Phi$. We also approximate the standardized asymmetric NIG distribution that arises in the continuous time case by the standard Normal, i.e., $F_{NIG}$ by $\Phi$, since it is only slightly positively skewed. Therefore we assume that the skewness is negligible. Moreover, recall that the shape parameter $\zeta = \alpha\delta$ and the excess kurtosis $\gamma$ of the symmetric NIG distribution are related by $\gamma = \frac{3}{\zeta}$. Thus, for each of the continuous time and discrete time cases, we obtain an easy to use Black-Scholes type formula for option pricing which gives correction that accounts for the access kurtosis.

In the continuous time case, the generalized or modified Black-Scholes formula for the log-symmetric NIG model (NIG-C) is given by
    \begin{align}
    C_0 &\approx S_0\Phi\left(\frac{\ln\big(\frac{S_0}{K}\big) + \mu_1 T}{\tilde{\sigma}_1\sqrt{T}} \right)
    - e^{-rT}K\Phi\left(\frac{\ln\big(\frac{S_0}{K}\big) + \mu T}{\tilde{\sigma}\sqrt{T}} \right), \label{eq:NIGC}
    \end{align}
where $\mu_1 = \mu + \sqrt{\frac{3}{3 - \gamma\sigma^2}}\tilde{\sigma}^2$, $\tilde{\sigma}_1^2 = \left(\sqrt{\frac{3}{3 - \gamma\sigma^2}}\right)^3\tilde{\sigma}^2$ and $\tilde{\sigma}^2 = 2(r - \mu) - \frac{\gamma}{3}(r - \mu)^2$, in which we have applied the fact that
    \begin{align*}
    \frac{\alpha^2}{\alpha^2 - 1} = \frac{\alpha^2\sigma^2}{\alpha^2\sigma^2 - \sigma^2} = \frac{\alpha\delta}{\alpha\delta - \sigma^2} = \frac{3}{3 - \gamma\sigma^2}.
    \end{align*}

Observe that, when $\gamma \rightarrow 0$, we have $\tilde{\sigma}^2 \rightarrow 2(r - \mu)$ and $\frac{3}{3 - \gamma\sigma^2} \rightarrow 1$. Consequently, the Black-Scholes formula is a special case of the generalized version (NIG-C) (\ref{eq:NIGC}) when $\gamma \rightarrow 0$ because
    \begin{align*}
    \mu_1 &= \mu + \sqrt{\frac{3}{3 - \gamma\sigma^2}}\tilde{\sigma}^2 \rightarrow \mu + 2(r - \mu),\\
    \tilde{\sigma}_1^2 &= \left(\sqrt{\frac{3}{3 - \gamma\sigma^2}}
    \right)^3\tilde{\sigma}^2 \rightarrow 2(r - \mu).\\
    \end{align*}
And let $2(r - \mu) = \sigma^2$, which is a constant as in the Black-Scholes model (Recall that under the risk-neutral measure $\mathbb{Q}$, the mean $\mu = r - \frac{\sigma^2}{2}$ and the volatility $\sigma$ is a constant), then by using these results and some simple manipulations, it is not hard to see that the generalized formula (NIG-C) (\ref{eq:NIGC}) is the exact Black-Scholes formula.

In the discrete time case, the modified Black-Scholes formula for log-symmetric NIG model (NIG-D) is given by
    \begin{align}
    C_0 &\approx S_0\Phi\left(\frac{\ln\big(\frac{S_0}{K}\big) + \bigg(r + \frac{3}{\gamma}\Big(1 - \sqrt{1 - \frac{\gamma\sigma^2}{3}}\Big)\bigg)N}{\sigma \sqrt{N}} \right) \nonumber\\
     & \hspace{10mm} - e^{-rN}K\Phi\left(\frac{\ln\big(\frac{S_0}{K}\big) + \bigg(r - \frac{3}{\gamma}\Big(1 - \sqrt{1 - \frac{\gamma\sigma^2}{3}}\Big)\bigg)N}{\sigma \sqrt{N}} \right). \label{eq:NIGD}
    \end{align}
It can be seen that the Black-Scholes formula is a limit of the generalized version (NIG-D) (\ref{eq:NIGD}) for every $N$ when $\gamma \rightarrow 0$ due to
    \[
    \frac{3}{\gamma}\left(1 - \sqrt{1 - \frac{\gamma\sigma^2}{3}}\right) \rightarrow \frac{\sigma^2}{2}.
    \]

An example of the option price formulae plotted against the expiration time $T$ using similar set of parameter values as in the log-symmetric VG model is given below (see Figure \ref{FigNIG}). Again, it is evident that the distinction between the modified Black-Scholes formulae (NIG-C and NIG-D) and BS is noticeable even for this moderate values of $\gamma$ (see Figure \ref{Fig1}). As in the previous model, the disagreement between NIG-C and BS formulae is greater than the disagreement between NIG-D and BS. The exact prices and percentage differences are represented in Table \ref{tab:tabNIG1}.

\begin{figure}[!htb]
    \centering
    \includegraphics [bb=70 90 715 507, scale=0.5]{graph2_NIG.eps}
    \caption{Option prices and percentage differences obtained by NIG-C, NIG-D and BS formulae for log-NIG distribution weekly returns, $S_0 = K = 10$, $r = 0.06$, $\sigma = 0.19$, $\mu = 0.03$, $\gamma = 4$}
    \label{FigNIG}
\end{figure}

\begin{table}[!htb]
\centering
\begin{tabular}{|l|c|c|c|c|c|c|}
  \hline

  \begin{minipage}{3cm}
  Time to maturity (weeks)
  \end{minipage}
  & 2 & 12 & 22 & 32 & 42 & 52 \\
  \hline \hline
  BS formula & 0.160 & 0.434 & 0.622 & 0.782 & 0.927 & 1.062 \\
  \hline
  NIG-D formula & 0.162 & 0.439 & 0.628 & 0.789 & 0.935 & 1.071 \\
  Percentage difference & 1.14 & 1.01 & 0.94 & 0.89 & 0.85 & 0.81 \\
  \hline
  NIG-C formula & 0.195 & 0.519 & 0.735 & 0.917 & 1.079 & 1.229 \\
  Percentage difference & 21.91 & 19.52 & 18.18 & 17.18 & 16.36 & 15.66\\
  \hline
\end{tabular}
\caption{Option prices and percentage differences obtained by NIG-C, NIG-D and BS formulae for log-NIG distribution weekly returns, $S_0 = K = 10$, $r = 0.06$, $\sigma = 0.19$, $\mu = 0.03$, $\gamma = 4$}
\label{tab:tabNIG1}
\end{table}

\newpage

\section{Appendix}

\begin{lemma}
\label{TransformIdentity}
Let $\nu$ and $\tilde{\nu}$ be two measure on $(0,+\infty)$ with finite second moments:
$$\kappa=\int_0^{\infty} y^2 \nu(dy)<+\infty\mbox{ and }\tilde{\kappa}=\int_0^{\infty} y^2 \tilde{\nu}(dy)<+\infty.$$
Let $\beta = \sqrt{\tilde{\kappa}/\kappa}$. If
\begin{equation}
\int_0^{\infty} \big(1 - \cos(\omega y)\big) \tilde{\nu}(dy) =
\int_0^{\infty} \big(1 - \cos(\beta\omega y)\big) \nu(dy),\quad \forall \omega > 0,\label{eq:NuTildeNu}
\end{equation}
then $\tilde{\nu}=\nu_\beta$, where $\nu_\beta$ is defined by $\int g(y)\nu_\beta(dy)=\int g(\beta y)\nu(dy)$.
\end{lemma}
\begin{proof}
To show that the two measures are identical we show that the Mellin
transforms of certain associated probability distributions are the same.

First we compute the Laplace transforms of the two sides of \eqref{eq:NuTildeNu}:
\begin{eqnarray*}
\lefteqn{\int_0^{\infty} e^{-\lambda\omega}\bigg(\int_0^{\infty} \big(1 - \cos(\beta\omega y)\big) \nu (dy)\bigg) d\omega} \\
& = & \int_0^{\infty}\bigg(\int_0^{\infty} e^{-\lambda\omega}\big(1 - \cos(\beta\omega y)\big)d\omega \bigg)\nu(dy)\ =\ \int_0^{\infty}\bigg(\frac{\beta^2y^2}{\lambda(\lambda^2 + \beta^2y^2)} \bigg)\nu(dy),
\end{eqnarray*}
and similarly for $\tilde{\nu}$. Here we have used the identity
$$\int_0^{\infty} e^{-\lambda\omega}\big(1 - \cos(\omega
y)\big)d\omega=\frac{y^2}{\lambda(\lambda^2 + y^2)}.$$
\eqref{eq:NuTildeNu} becomes
        \begin{equation}
        \int_0^{\infty}\frac{y^2}{\lambda^2 + y^2} \tilde{\nu}(dy) =
        \int_0^{\infty}\frac{\beta^2y^2}{\lambda^2 + \beta^2y^2} \nu(dy), \;\;\forall \lambda > 0. \label{eq:AfterFubini2}
        \end{equation}
Consider now the probability measures with support on $(0,+\infty)$, $\tilde{n}(dy)=\frac1{\tilde{\kappa}}y^2\tilde{\nu}(dy)$
and $n(dy)=\frac1{\tilde{\kappa}}y^2\nu_\beta(dy)$.
Then the positive random variables $X$ and $\widetilde{X}$ with respective distributions
$n$ and $\tilde n$ satisfy, for all $\lambda>0$,
$$\mathbb{E}\left[\frac{1}{\lambda^2+X^2}\right]=\mathbb{E}\left[\frac{1}{\lambda^2+\widetilde{X}^2}\right].$$
In other words, the Mellin transforms of $X^2$ and $\widetilde{X}^2$are equal, which in turn implies that the laws
of $X$ and $\widetilde X$ are the same, that $\tilde{n}=n$, and consequently that
$\tilde{\nu}=\nu_\beta$.
\end{proof}

\begin{lemma}
    \label{DominatedConvergence}

    \[
    \lim_{v\to\infty}\int_0^{\infty} \frac{1}{v}\Big[ \big(1 - \cos(y\sqrt{2v}/\tilde{\sigma})\big)e^{\phi(y)} -
    \big(1 - \cos(y\sqrt{2v}/\sigma)\big) \Big] \nu(dy)=0.
    \]

\end{lemma}

\begin{proof}
Let
$$f_v(y) = \frac{1}{v}\Big[ \big(1 - \cos(y\sqrt{2v}/\tilde{\sigma})\big)e^{\phi(y)} - \big(1 - \cos(y\sqrt{2v}/\sigma)\big) \Big].$$
Clearly, for any fixed $y$, $\lim_{v \rightarrow \infty} f_v(y) = 0$.
Using the inequality $1 - \cos(x)\le x^2/2$ we
obtain
$$|f_v(y)|\leq\frac{y^2}{\tilde{\sigma}^2}e^{\phi(y)} + \frac{y^2}{\sigma^2} = G(y).$$
$G(y)$ is integrable with respect to $\nu$, since the L\'evy measures
$\tilde{\nu}$ and $\nu$ satisfy
    \[
    \int_{\mathbb{R}}(1 \wedge y^2) \tilde{\nu}(dy) < \infty, \hspace{8mm} \mbox{and} \hspace{8mm} \int_{\mathbb{R}}(1 \wedge y^2) \nu(dy) < \infty;
    \]
and   existence of variance implies
    \[
    \int_{|y| > 1} y^2 \tilde{\nu}(dy) < \infty, \hspace{8mm} \mbox{and} \hspace{8mm} \int_{|y| > 1} y^2 \nu(dy) <
    \infty.
    \]
The result follows by dominated convergence.

\end{proof}

\section*{Acknowledgements}
This research was supported by the
Australian Research Council grant DP0988483.




\begin{thebibliography}{10}

\bibitem{BarndorffNielsen95}
Barndoff-Nielsen, O. E.
\newblock {\em Normal Inverse Gaussian Processes and the Modelling of Stock Returns}.
\newblock Research Report 300, Department of Theoretical Statistics, university of Aarhus.


\bibitem{Benhamou02}
Benhamou, E.
\newblock {\em Option Pricing with Levy Process}.
\newblock EconWPA, Finance, 2002,\\
\newblock http://129.3.20.41/eps/fin/papers/0212/0212006.pdf.


\bibitem{BlattbergGonedes74}
Blattberg, R. C and Gonedes, N. J.
\newblock {\em A Comparison of the Stable and Student Distribution as Statistical Models for Stock Prices}.
\newblock J. Business, {\bf 47}, 1974, pp. 244 - 280.

\bibitem{Chan99}
Chan, T.
\newblock {\em Pricing Contingent Claims on Stocks Driven by Levy Processes}.
\newblock Ann. Applied Probability, {\bf 9}:2, 1999, 504 - 528.


\bibitem{CarrWu04}
Carr, P. and Wu, L.
\newblock {\em Time-changed L\'evy Processes and Option Pricing}.
\newblock J. Financial Economics, {\bf 71}, 2004, pp. 113 - 141.

\bibitem{ContTankov04}
Cont, R. and Tankov, P.
\newblock{\em Financial Modelling with Jump Processes}
\newblock Chapman and Hall, 2004, New York.


\bibitem{EberleinJacob97}
Eberlein, E. and Jacod, J.
\newblock {\em On the Range of Option Prices}.
\newblock Finance and Stochastic, {\bf 1}, 1997, pp. 131 - 140.


\bibitem{ElliottMadan98}
Elliott, R. and Madan, D.
\newblock {\em A Discrete Time Equivalent Martingale Measures}.
\newblock Mathematical Finance {\bf 8}:2, 1998, pp. 127 - 152.

\bibitem{FajardoMordecki06}
Fajardo, J. and Mordecki, E.
\newblock {\em Symmetry and Duality in L\'evy Markets}.
\newblock Quantitative Finance, {\bf 6}:3, 2006, pp. 219 - 227.

\bibitem{Fama65}
Fama, E. F.
\newblock {\em The Behaviour of Stock Market Prices}.
\newblock J. Business, {\bf 37}, 1965, pp. 34 - 105.

\bibitem{FangOthers}
Fang, K.~T., Kotz, S. and Ng, K.~W.
\newblock {\em Symmetric Multivariate and Related Distributions}. Chapman \& Hall, 1990, London.

\bibitem{FollmerSchweizer91}
Follmer, H. and Schweizer, M.
\newblock {\em Hedging of Contingent Claims under Incomplete Information}.
\newblock In M. H. A. Davis and R. J. Elliot (ed.): Applied Stochastic Analysis,
\newblock Gordon and Breach, 1991, pp. 389 - 414.

\bibitem{Frittelli00}
Frittelli, M.
\newblock {\em The Minimal Entropy Martingale Measures and the Valuation Problem in Incomplete Markets}.
\newblock Mathematical Finance, {\bf 10}, 2000, pp. 39 - 52.

\bibitem{FujiwaraMiyahara03}
Fujiwara, T. and Miyahara, Y.
\newblock {\em The Minimal Entropy Martingale Measures for Geometric L\'evy Processes}.
\newblock Finance and Stochastics, {\bf 7}:4, 2003, pp. 509 - 531.

\bibitem{GerberShiu94}
Gerber, H. and Shiu, E.
\newblock {\em Option Pricing by Esscher Transform}.
\newblock Transactions Soc. Actuaries, {\bf 46}, 1994, pp. 99 - 191.

\bibitem{GemanOthers95}
Geman, H., El Karoui, N. and Rochet, J-C.
\newblock {\em Changes of Numeraire, Changes of Probability Measure and Option Pricing}.
\newblock J. Appl. Prob., {\bf 32}, 1995, pp. 443 - 458.


\bibitem{GollRuschendorf01}
Goll, T. and R\"{u}schendorf, L.
\newblock {\em Minimax and Minimal Distance Martingale Measure and Their Relationship to Portfolio Optimization}.
\newblock Finance and Stochastics, {\bf 5}:4, 2001, pp. 557 - 581.

\bibitem{Hurlimann95}
H\"{u}rlimann, W.
\newblock {\em Is There a Rational Evidence for an Infinite Variance Asset Pricing Model?}.
\newblock Proceedings of the 5-th International AFIR Colloquium, 1995.

\bibitem{Hurlimann01}
H\"{u}rlimann, W.
\newblock {\em Financial Data Analysis with Two Symmetric Distributions}.
\newblock ASTIN Bulletin, {\bf 31}, 2001, pp. 187 - 211.


\bibitem{JohnsonKotzBala94}
Johnson, N., Kotz, S. and Balakrishnan, N.
\newblock {\em Continuous Univariate Distributions}. Vol 1, Wiley, 1994, New York.

\bibitem{KallsenShiryaev02}
Kallsen, J. and Shiryaev, A. N.
\newblock {\em The Cumulant Process and Esscher's Change of Measure}.
\newblock Finance and Stochastics, {\bf 6}:4, 2002, pp. 397 - 428.

\bibitem{Klebaner2005}
Klebaner, F.~C.
\newblock {\em Introduction to Stocahstic Calculus with Applications, 2nd Ed}.
\newblock Imperial College Press, 2005.

\bibitem{KlebanerLandsman07}
Klebaner, F.~C. and Landsman, Z.
\newblock {\em Option pricing for Log-Symmetric Distributions of Returns}.
\newblock Methodol. Comput. Appl. Probab., DOI 10.1007/s11009-007-9038-2, 2007.

\bibitem{KuchlerTappe08}
K\"{u}chler, U. and Tappe, S.
\newblock {\em Bilateral Gamma Distributions and Processes in Financial Mathematics}.
\newblock Stoch. Proc. Appl., {\bf 118}, 2008, p. 261 - 283.

\bibitem{MadanCarrChang98}
Madan, D., Carr, P. and Chang, E.
\newblock {\em The Variance Gamma Process and Option Pricing}.
\newblock European Finance Review, {\bf 2}, 1998, pp. 79 - 105.


\bibitem{MadanMilne91}
Madan, D. and Milne, F.
\newblock {\em Option Pricing with Variance Gamma Martingale Components}.
\newblock Math. Finance, {\bf 1}, 1991, pp. 39 - 55.


\bibitem{McDonald96}
McDonald, J. B.
\newblock {\em Probability Distributions for Financial Models}.
\newblock Hanbook of Statistics, {\bf 14}, 1996, pp. 427 - 461.

\bibitem{Mendelbrot63}
Mendelbrot, B.
\newblock {\em The Variation of Certain Speculative Prices}.
\newblock J. Business, {\bf 36}, 1963, pp. 394 - 419.

\bibitem{Mendelbrot67}
Mendelbrot, B.
\newblock {\em The Variation of some other Speculative Prices}.
\newblock J. Business, {\bf 40}, 1967, pp. 393 - 413.

\bibitem{Miyahara99}
Miyahara, Y.
\newblock {\em Minimal Entropy Martingale Measures of Jump Type Price Processes in Incomplete Assets Markets}.
\newblock Asian-Pacific Financial Markets, {\bf 6}:2, 1999, pp. 97 - 113.

\bibitem{Rydberg97}
Rydberg, T. H.
\newblock {\em The Normal Inverse Gaussian L\'evy Process: Simulation and Approximaltion}.
\newblock Commun. Statist.-Stochastic Models, {\bf 13}:4, 1997, pp. 887 - 910.


\bibitem{Sato1999}
Sato, K. I.
\newblock {\em L\'evy Processes and Infinitely Divisible Distributions}.
\newblock Cambridge University Press, 1999, Cambridge.

\bibitem{Schoutens03}
Schoutens, W.
\newblock {\em L\'evy processes in Finance: Pricing Financial Derivatives}.
\newblock Wiley, 2003, Chichester.

\bibitem{Schweizer96}
Schweizer, M.
\newblock {\em Approximation Pricing and the Variance-Optimal Martingale Measure}.
\newblock Ann. Probability, {\bf 24}:1, 1996, pp. 206 - 236.

\bibitem{ShiryaevKruzhilin1999}
Shiryaev, A. N. and Kruzhilin, N.
\newblock {\em Essentials of Stochastic Finance: Facts, Models, Theory}.
\newblock World Scientific, 1999.



\end{thebibliography}
\end{document}